%% file: main.tex
\documentclass[11pt,letterpaper]{article}

\input{config/packages}
\input{config/commands}

\title{\LARGE \bf Reach-SDP: Reachability Analysis of Closed-Loop Systems with Neural Network Controllers via Semidefinite Programming}
\author{Haimin Hu, Mahyar Fazlyab, Manfred Morari, and George J. Pappas
\thanks{Work supported by the NSF under grants {DARPA Assured Autonomy and NSF CPS 1837210}. The authors are with the Department of Electrical and Systems Engineering, University of Pennsylvania. Email: \{haiminhu,mahyarfa,morari,pappasg\}@seas.upenn.edu.}}

\begin{document}
\pagestyle{plain}
\date{}
\maketitle

\begin{abstract}
	There has been an increasing interest in using neural networks in closed-loop control systems to improve performance and reduce computational costs for on-line implementation.
	However, providing safety and stability guarantees for these systems is challenging due to the nonlinear and compositional structure of neural networks.
	In this paper, we propose a novel forward reachability analysis method for the safety verification of linear time-varying systems with neural networks in feedback interconnection.
	Our technical approach relies on abstracting the nonlinear activation functions by quadratic constraints, which leads to an outer-approximation of forward reachable sets of the closed-loop system.
	We show that we can compute these approximate reachable sets using semidefinite programming.
	We illustrate our method in a quadrotor example, in which we first approximate a nonlinear model predictive controller via a deep neural network and then apply our analysis tool to certify finite-time reachability and constraint satisfaction of the closed-loop system.
\end{abstract}

\input{contents/introduction}
\input{contents/prob-form}
\input{contents/abs-QC}
\input{contents/reach-SDP}
\input{contents/experiments}
\input{contents/conclusion}
\input{contents/appendix}

%


\bibliographystyle{plain}
\bibliography{main}

\newpage
\appendix

\end{document}

%% file: config/packages.tex
\usepackage{graphicx, amsmath,amsthm,amssymb,url}

\usepackage[labelfont=bf]{subcaption}
\usepackage[labelfont=bf]{caption}
\usepackage{algorithm,algcompatible,mathtools}
\usepackage[multidot]{grffile}
\usepackage[normalem]{ulem} 

\usepackage{fullpage,etoolbox}
\usepackage{color}
\usepackage[numbers,sort,compress]{natbib}
\usepackage{authblk}
\usepackage{mathtools}

\usepackage[margin=1.0in]{geometry}

\usepackage{caption} 
\usepackage{caption} 


%% file: config/commands.tex
\newtheorem{theorem}{Theorem}
\newtheorem{lemma}{Lemma}    

\newtheorem{remark}{Remark}

\newtheorem{proposition}{Proposition}    
    
\newtheorem{definition}{Definition}

\newcommand{\R}{\mathbb{R}}

\newcommand{\bx}{\mathbf{x}}

\DeclareMathOperator*{\minimize}{minimize \quad}

\def\bx{\underline{\mathbf{x}}}

\def\bx{\mathbf{x}}

\def\R{\mathbb{R}}
\def\Sym{\mathbb{S}}
\def\proj{\operatorname{Proj}_{\mathcal{U}_t}}
\def\nnc{\pi}

\def\Reach{\mathcal{R}}
\def\ReachBar{\bar{\mathcal{R}}}
\def\init{\mathcal{X}_0}
\def\goal{\mathcal{G}}
\def\avoid{\mathcal{A}}

\def\Ein{E_{\rm{in}}}

\def\Min{M_{\rm{in}}}

\def\cMin{\mathcal{M}_{\rm{in}}}

\def\Emid{E_{\rm{mid}}}
\def\Mmid{M_{\rm{mid}}}
\def\cMmid{\mathcal{M}_{\rm{mid}}}

\def\Eout{E_{\rm{out}}}

\def\Mout{M_{\rm{out}}}

\def\mpc{\text{MPC}}

%% file: contents/introduction.tex
\section{Introduction}
Deep neural networks (DNN) have seen renewed interest in recent years due to the proliferation of data and access to more computational power. In autonomous systems, DNNs are either used as feedback controllers \cite{zhang2016learning,hertneck2018learning}, motion planners \cite{qureshi2019motion}, perception modules, or  end-to-end controllers \cite{pan2017agile,bojarski2016end}. Despite their high performance, DNN-driven autonomous systems lack formal safety and stability guarantees. Indeed, recent studies show that DNNs can be vulnerable to small perturbations or adversarial attacks \cite{papernot2016limitations,kurakin2016adversarial}. This issue is more pronounced in closed-loop systems, as a small perturbation in the loop can dramatically change the behavior of the closed-loop system over time. Therefore, it is of utmost importance to develop tools for verification of DNN-driven control systems. The goal of this paper is to develop a methodology, based on semidefinite programming, for safety verification and reachability analysis of linear dynamical systems in feedback interconnection with DNNs.

Safety verification or reachability analysis aims to show that starting from a set of initial conditions, a dynamical system cannot evolve to an unsafe region in the state space. Methods for reachability analysis can be categorized into exact (complete) or approximate (incomplete), which compute the reachable sets exactly and approximately, respectively. Verification of dynamical systems has been extensively studied in the past \cite{kurshan2014computer,bemporad1999verification,bemporad2000optimization,prajna2004safety,tomlin2003computational}. More recently, the problem of output range analysis of neural networks has been addressed in \cite{huang2017safety,lomuscio2017approach,ehlers2017formal,katz2017reluplex,raghunathan2018semidefinite,fazlyab2019safety,fazlyab2019efficient,kolter2017provable,dutta2017output}, mainly motivated by robustness analysis of DNNs against adversarial attacks \cite{kurakin2016adversarial}. Compared to these bodies of work, verification of closed-loop systems with neural network controllers has been less explored \cite{huang2019reachnn,ivanov2019verisig,dutta2019reachability}. In \cite{ivanov2019verisig}, a method for verification of sigmoid-based neural networks in feedback with a hybrid system is proposed, in which the neural network is transformed into a hybrid system and then a standard verification tool for hybrid systems is invoked. In \cite{huang2019reachnn}, a new reachability analysis approach based on Bernstein polynomials is proposed that can verify DNN-controlled systems with Lipschitz continuous activation functions. Dutta et al. \cite{dutta2019reachability} use a flow pipe construction scheme to over approximate the reachable sets. A piecewise polynomial model is used to provide an approximation of the input-output mapping of the controller and an error bound on the approximation. This approach, however, is only applicable to Rectified Linear Unit (ReLU) activation functions.

\textit{Contributions.} In this paper, we propose a semidefinite program (SDP) for reachability analysis of linear time-varying dynamical systems in feedback interconnection with neural network controllers equipped with a projection operator, which projects the output of the neural network (the control action) onto a specified set of control inputs.
Our technical approach relies on abstracting the nonlinear activation functions as well as the projection operator by quadratic constraints \cite{fazlyab2019safety}, which leads to an outer-approximation of forward reachable sets of the closed-loop system. We show that we can compute these approximate reachable sets using semidefinite programming. 
Our approach can be used to analyze control policies learned by neural networks in, for example, model predictive control (MPC)~\cite{chen2018approximating} and constrained reinforcement learning~\cite{wen2018constrained}.
%
%
%
To the best of our knowledge, our result is the first convex-optimization-based method for reachability analysis of closed-loop systems with neural networks in the loop. We illustrate the utility of our approach in two numerical examples, in which we certify finite-time reachability and constraint satisfaction of a double integrator and a quadrotor.

\subsection{Notation and Preliminaries}
We denote the set of real numbers by $\R$, the set of real $n$-dimensional vectors by $\R^n$, the set of $m \times n$-dimensional
matrices by $\R^{m \times n}$, and the $n$-dimensional identity matrix by $I_n$.
We denote by $\Sym^n$, $\Sym^n_+$, and $\Sym^n_{++}$ the sets of $n$-by-$n$ symmetric, positive semidefinite, and positive definite matrices, respectively.
For $A \in \R^{m \times n}$, the inequality $A \geq 0$ means all entries of $A$ are non-negative.
For $A \in \Sym^n$, the inequality $A \succeq 0$ means $A$ is positive semidefinite.
\begin{definition}[Sector-bounded nonlinearity]
	A nonlinear function $\varphi: \R \rightarrow \R$ is sector-bounded on $[\alpha, \beta]$, where $0 \leq \alpha \leq \beta$, if the following inequality holds for all $x \in \R$,
	\begin{align}
	\alpha \leq \frac{ \varphi(x)}{x} \leq \beta,
	\end{align}
	which can be equivalently expressed as the quadratic inequality
	\begin{equation}
	\label{eq::QC::sector_bounded}
	\begin{bmatrix}
	x \\ \varphi(x)
	\end{bmatrix}^\top \begin{bmatrix}
	-2\alpha \beta & \alpha+\beta \\ \alpha+\beta & -2
	\end{bmatrix}\begin{bmatrix}
	x \\ \varphi(x)
	\end{bmatrix} \geq 0.
	\end{equation}
\end{definition}

\begin{definition}[Slope-restricted nonlinearity]
	A nonlinear function $\varphi: \R \rightarrow \R$ is slope-restricted on $[\alpha, \beta]$, where $0 \leq \alpha \leq \beta$, if for any pairs of $(x, \varphi(x))$ and $\left(x^{\star}, \varphi\left(x^{\star}\right)\right)$,
	\begin{align}
	\alpha \leq \frac{\varphi(x)-\varphi(x^\star)}{x-x^\star} \leq \beta,
	\end{align}
	which can be equivalently expressed as
	\begin{equation}
	\label{eq::QC::slope_res}
	\begin{bmatrix}
	x\!-\!x^\star \\ \varphi(x)\!-\!\varphi(x^\star)
	\end{bmatrix}^\top \begin{bmatrix}
	-2\alpha \beta & \alpha+\beta \\ \alpha+\beta & -2
	\end{bmatrix}\begin{bmatrix}
	x\!-\!x^\star \\ \varphi(x)\!-\!\varphi(x^\star)
	\end{bmatrix} \geq 0.
	\end{equation}
\end{definition}

%% file: contents/prob-form.tex
\section{Problem Formulation}

\subsection{Neural Network Control System}
We consider a discrete-time linear time-varying system
\begin{equation}
\label{eq::LTV}
P: \ x_{t+1} = A_t x_t + B_t u_t + c_t,
\end{equation}
where $x_t \in \R^{n_x}$, $u_t \in \R^{n_u}$ are the state and control vectors, and $c_t \in \R^{n_x}$ is an exogenous input.
We assume that the system \eqref{eq::LTV} is subject to input constraints,
\begin{equation}
\label{eq::constr}
u_t \in \mathcal{U}_t, \ t=0,1,\cdots.
\end{equation}
which represent, for example, actuator limits that are naturally satisfied or hard constraints that must be satisfied by a control design specification. A specialization that we consider in this paper is input box constraint $\mathcal{U}_t = \{u_t \mid \underline{u}_t \leq u_t \leq \bar{u}_t\}$.
Furthermore, we assume a state-feedback controller $\nnc\left(x_t\right): \R^{n_x} \rightarrow \R^{n_u} $ parameterized by a multi-layer feed-forward fully-connected neural network. The map $x \mapsto \nnc(x)$ is described by the following equations,
\begin{equation}
\label{eq::NN}
\begin{aligned}
    x^0 &= x \\
    x^{k+1}&=\phi(W^{k} x^{k}+b^{k}) \quad k=0, \cdots, \ell-1 \\
    \nnc(x) &= W^{\ell} x^{\ell}+b^{\ell},
\end{aligned}
\end{equation}
where $W^k \in \R^{n_{k+1} \times n_{k}}$, $b^{k} \in \R^{n_{k+1}}$ are the weight matrix and bias vector of the $(k+1)$-th layer.
The nonlinear activation function $\phi(\cdot)$ is applied component-wise to the pre-activation vectors, i.e.,
\begin{equation}
    \phi(x) \coloneqq [ \varphi(x_{1}) \cdots \varphi(x_{d})]^\top, \ x\in \R^d,
\end{equation}
where $\varphi: \R \rightarrow \R$ is the activation function of each individual neuron.
Common choices include ReLU, sigmoid, tanh,
leaky ReLU, etc. In this paper, we consider ReLU activation functions in our technical derivations but we can address other activation functions following the framework of \cite{fazlyab2019safety}.
To ensure that output of neural network respects the input constraint, we consider a projection operator in the loop and define the control input as
\begin{equation}
\label{eq::sat_control}
u_t = \proj \left( \nnc(x_t) \right),
\end{equation}
%
We denote the closed-loop system with dynamics \eqref{eq::LTV} and the projected neural network control policy \eqref{eq::sat_control} by
\begin{equation}
\label{eq::closed_loop_sys}
    x_{t+1} = f_\pi \left( x_t \right),
\end{equation}
which is a non-smooth nonlinear system because of the presence of the nonlinear activation functions in the neural network and the projection operator.
%
\begin{remark}
	In this paper, we only consider box  constraints for the input 
	and leave more sophisticated constraints such as polytopes to future work. These constraints are useful in, for example, MPC design problems. In~\cite{hertneck2018learning} a robust MPC controller is approximated by a neural network equipped with a projection operator to ensure satisfaction of polytopic input constraints.
\end{remark}
\noindent For the closed-loop system \eqref{eq::closed_loop_sys} subject to the input constraint in \eqref{eq::constr}, we denote by $\Reach_t(\init)$ the forward reachable set at time $t$ from a given set of initial conditions $\init \subseteq \R^{n_x}$, which is defined by the recursion
\begin{equation}
\label{eq::one_step_reach}
\begin{aligned}
\Reach_{t+1}&(\init) \coloneqq f_{\pi}(\Reach_{t}(\init)), \ \Reach_0(\init) = \init. 
\end{aligned}
\end{equation}
as illustrated in Figure \ref{fig::closed_loop_reach}.
%
\begin{figure}[!hbtp]
    \centering
    \includegraphics[width=0.6\columnwidth]{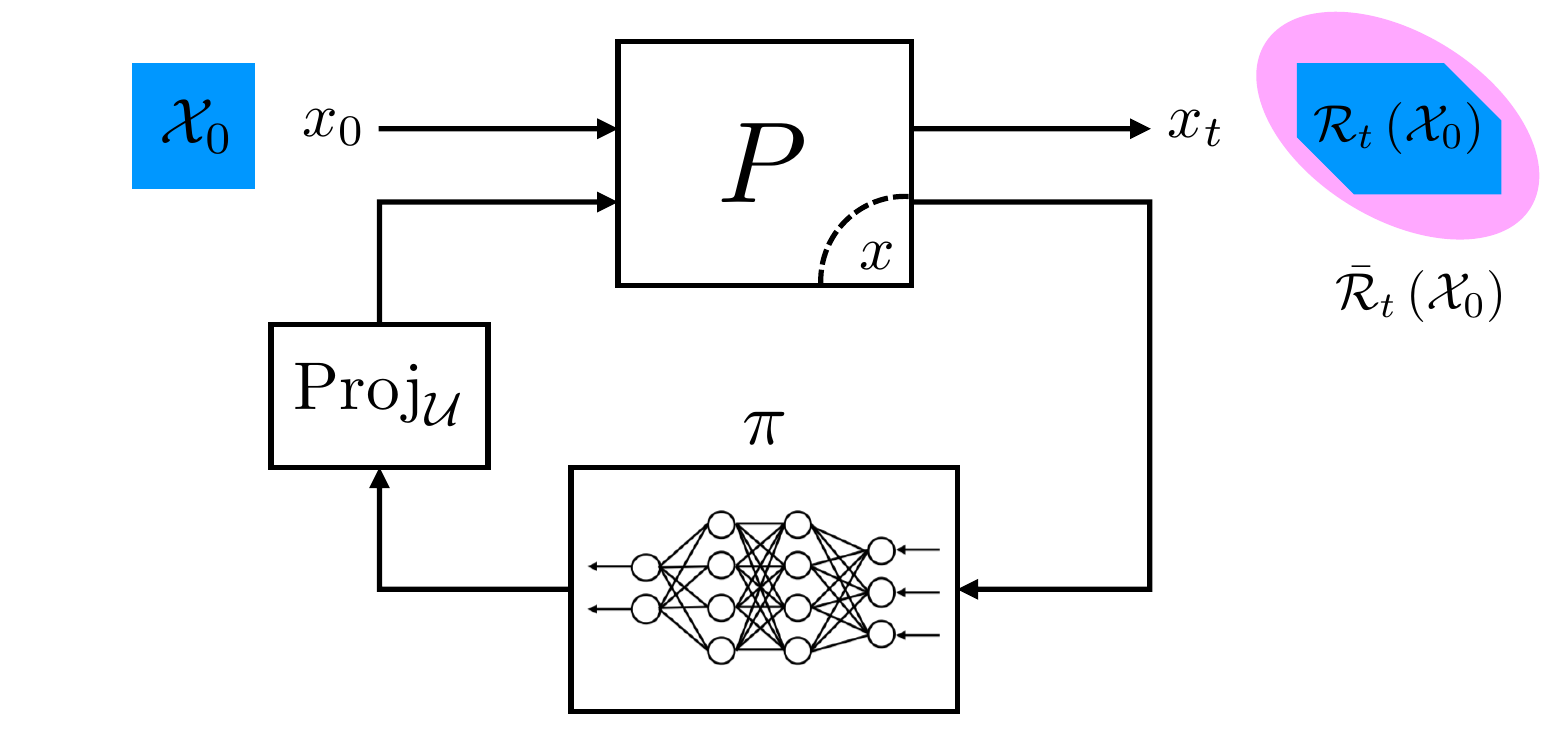}
    \caption{\label{fig::closed_loop_reach} An illustration of closed-loop reachability with the initial set $\init$, the $t$-step forward reachable set $\Reach_t(\init)$, and its over-approximation $\ReachBar_t(\init)$ shown in magenta. } 
\end{figure}

\subsection{Finite-Time Reach-Avoid Verification Problem}
In this paper, we are interested in verifying the finite-time reach-avoid properties of the closed-loop system \eqref{eq::closed_loop_sys}.
More specifically, given a goal set $\goal \subseteq \R^{n_x}$ and a sequence of avoid sets $\avoid_t \subseteq \R^{n_x}$, we would like to test if all initial states $x_0 \in \init$ of the closed-loop system \eqref{eq::closed_loop_sys} can reach $\goal$ in a finite time horizon $N \geq 0$, while avoiding $\avoid_t$ for all $t = 0,\cdots,N$.
This is equivalent to test if,
\begin{subequations}
\label{eq::reach_avoid}
\begin{align}
    \label{eq::reach_avoid_1} 
    & \Reach_N(\init) \subseteq \goal \\
    \label{eq::reach_avoid_2}  
    & \Reach_t(\init) \cap \avoid_t = \emptyset, \ \forall t = 0,\cdots,N
    \end{align}
\end{subequations}
holds true for \eqref{eq::closed_loop_sys}.
%
There exist efficient methods~\cite{blanchini2008set} and software implementations~\cite{MPT3} for testing set inclusion \eqref{eq::reach_avoid_1} and set intersection \eqref{eq::reach_avoid_2}.
However, computing exact reachable sets for the nonlinear closed-loop system \eqref{eq::closed_loop_sys} is, in general, computationally intractable.
Therefore, we resort to finding outer-approximations of the closed-loop reachable sets, $\ReachBar_t(\init) \supseteq \Reach_{t}(\init)$, and use them to test if,
\begin{subequations}
\label{eq::reach_avoid_approx}
\begin{align}
    \label{eq::reach_avoid_approx_1} 
    & \ReachBar_N(\init) \subseteq \goal \\
    \label{eq::reach_avoid_approx_2}  
    & \ReachBar_t(\init) \cap \avoid_t = \emptyset, \ \forall t = 0,\cdots,N
    \end{align}
\end{subequations}
is true. Note that \eqref{eq::reach_avoid_approx} are sufficient conditions for \eqref{eq::reach_avoid}. To obtain meaningful certificates, we  want the approximations $\ReachBar_t(\init)$ to be as tight as possible.
Thus our goal is to compute the tightest outer-approximations of the $t$-step reachable sets.
This can be addressed, for example, by solving the following optimization problem,
\begin{equation}
\begin{aligned}
\label{eq::min_vol}
\operatorname{minimize} \ &\operatorname{Volume}\left(\ReachBar_t(\init)\right) \\
\text{subject to} \ &\Reach_{t}(\init) \subseteq \ReachBar_t(\init).
\end{aligned}
\end{equation}
The solution to the above problem is the minimum-volume outer-approximation of the $t$-step reachable set of the closed-loop system \eqref{eq::closed_loop_sys}.
In the following sections, we will derive a convex relaxation to the optimization problem \eqref{eq::min_vol}.

%% file: contents/abs-QC.tex
\section{Problem Abstraction via Quadratic Constraints}
\label{sec::QC}
This section focuses on estimating the one-step forward reachable set $\ReachBar(\init)$.
The main idea is to replace the original closed-loop system $f_\pi$ with an abstracted system $\tilde{f}_\pi$ in the sense that $\tilde{f}_\pi$ over-approximates the output of the original system for any given initial set $\init$, i.e. $\Reach(\init) = f_\pi(\init) \subseteq \tilde{f}_\pi(\init) = \ReachBar(\init)$.
Based on this abstracted system, we can then compute $\ReachBar(\init)$ via semidefinite programming (SDP).
In the following, we will develop such an abstraction using the framework of Quadratic Constraints (QCs).

\subsection{Initial Set}
We begin with a formal definition of QCs for sets \cite{fazlyab2019safety}.
\begin{definition}[Quadratic Constraints]
	Let $\mathcal{X} \subset \R^{d}$ be a nonempty set and $\mathcal{Q} \subset \Sym^{d+1}$ be the set of all symmetric, but possibly indefinite matrices $Q$ such that the inequality
	\begin{equation}
	\label{eq::QC::def}
	\begin{bmatrix}
	x \\
	1
	\end{bmatrix}^{\top} Q
	\begin{bmatrix}
	x \\
	1
	\end{bmatrix} \geq 0,
	\end{equation}
	holds for all $x \in \mathcal{X}$.
	Then we say $\mathcal{X}$ satisfies the QC defined by $\mathcal{Q}$. The vector $\begin{bmatrix}
	x^\top \ 1 \end{bmatrix}^\top$ is called the basis of this QC.
\end{definition}
\noindent Here, for each fixed $Q$, the set of $x$'s satisfying \eqref{eq::QC::def} is a superset of $\mathcal{X}$. Indeed, we have that
\begin{equation}
\mathcal{X} \subseteq \bigcap_{Q \in \mathcal{Q}}\left\{x \in \mathbb{R}^{d} \,\middle\vert\, \begin{bmatrix}
x \\
1
\end{bmatrix}^{\top} Q 
\begin{bmatrix}
x \\
1
\end{bmatrix} \geq 0\right\}.
\end{equation}
%
In this paper, we mainly use polytopes and ellipsoids as the initial set $\init$.
Nonetheless, as addressed by \cite{fazlyab2019safety}, other types of sets such as hyper-rectangles and zonotopes are also applicable in this setting.

\begin{proposition}[QC for polytope]
Suppose the initial set $\init$ is a polytope defined by $\init = \left\{x \in \mathbb{R}^{n_{x}} \mid A x \leq b\right\}$.
Then $\init$ satisfies the QC defined by
\begin{equation}
\label{eq::QC::poly}
\mathcal{P}=\left\{P \,\middle\vert\, P=\begin{bmatrix}
A^{\top} \Gamma A &-A^{\top} \Gamma b \\
-b^{\top} \Gamma A &b^{\top} \Gamma b
\end{bmatrix}, \Gamma \in \mathbb{S}^{m}, \Gamma \geq 0\right\},
\end{equation}
where $m$ is the dimension of $b$.
The basis is $\begin{bmatrix}
x^\top \ 1 \end{bmatrix}^\top$.
\end{proposition}

\begin{proposition}[QC for ellipsoid]
Suppose the initial set $\init$ is an ellipsoid defined by $\init = \left\{ x \in \R^{n_x} \mid \| Ax + b \|_2 \leq 1 \right\}$, where $A \in \Sym^{n_x}$ and $b \in \R^{n_x}$.
Then $\init$ satisfies the QC defined by
\begin{equation}
\label{eq::QC::ellips}
\mathcal{P}=\left\{P \,\middle\vert\, P=\mu\begin{bmatrix}
-A^{\top} A & -A^{\top} b \\
-b^{\top} A & 1-b^{\top} b
\end{bmatrix}, \ \mu \geq 0\right\},
\end{equation}
with basis $\begin{bmatrix} x^\top \ 1 \end{bmatrix}^\top$.
\end{proposition}

\noindent To this end, we have over-approximated the initial set $\init$ and represented it as QCs.
We will see in Section \ref{sec::SDP} that the matrix $P \in \mathcal{P}$ appears as a decision variable in the SDP.
It provides an extra degree of freedom towards a less conservative estimation of the reachable set $\ReachBar(\init)$.

\subsection{Reachable Set}
\label{sec::eq::QC::reach_set}
In order to facilitate the relaxation of the volume-minimization problem in \eqref{eq::min_vol}, we assume the candidate set $\ReachBar(\init)$ that over-approximates $\Reach(\init)$ is represented by the intersection of finitely many \textit{quadratic inequalities}:
\begin{equation}
\begin{aligned}
\label{QI::reach_set}
\ReachBar(\init) &=
\! \bigcap_{i=1}^{m} \! &\left\{ x \in \R^{n_x} \,\middle\vert\,  \begin{bmatrix}
x \\
1
\end{bmatrix}^{\top} \! S_{i} \! \begin{bmatrix}
x \\
1
\end{bmatrix} \leq 0 
\right\},
\end{aligned}
\end{equation}
where matrices $S_i \in \Sym^{2n_x+1}$ capture the shape and volume of the reachable set and will appear as decision variables in the SDP problem.
Typically, \eqref{QI::reach_set} is able to describe polytopic and ellipsoidal sets, which we discuss in detail now.

\subsubsection{Polytopic reachable set}
If the reachable set is to be over-approximated by a polytope, i.e. $\ReachBar(\init) = \left\{ x \in \R^{n_x} \mid A f_\pi(x) \leq b \right\}$, then,
\begin{equation}
\label{QI::reach_set_poly}
S_i =
\left[\begin{array}{ccc}
0 & a_i^\top \\
a_i & -2b_i
\end{array}\right],
\end{equation}
where $a_i \in \R^{1 \times n_x}$ is the $i$-th row of $A \in \R^{m \times n_x}$ and $b_i \in \R$ is the $i$-th entry of $b \in \R^m$.
Here, we require the $A$ matrix, which determines the orientation of each facet of the polytope, to be given while we leave $b_i$'s as decision variables in the SDP problem.

\subsubsection{Ellipsoidal reachable set}
If the reachable set is to be described by an ellipsoid, i.e. $\ReachBar(\init) = \left\{ x \in \R^{n_x} \mid \| A f_\pi(x) + b \|_2 \leq 1 \right\}$, then,
\begin{equation}
\label{QI::reach_set_ellips}
S_i = S =
\left[\begin{array}{cc}
A^\top A & A^\top b \\
b^{\top} A & b^{\top} b-1
\end{array}\right],
\end{equation}
where $A \in \Sym^{n_x}$ and $b \in \R^{n_x}$ are decision variables describing the center, orientation, and volume of the ellipsoid.

\begin{remark}
For polytopic reachable sets, if the facets $a_i$ are properly chosen, the resulting outer-approximations can be very tight, as we will see in Section \ref{sec::exp::DI}.
However, finding facets of a higher dimensional polytope can be prohibitively challenging.
Ellipsoidal reachable sets scale better and therefore are more suitable in this case.
\end{remark}

\subsection{ReLU Neural Networks with Projection}
In this subsection, we show how to abstract the nonlinear activation functions of the neural network by QCs. We will focus on the ReLU function throughout the paper.
Other types of activation functions such as sigmoid and tanh can also be represented by QCs. See \cite{fazlyab2019safety} for details. Recall that the input constraint sets are $\mathcal{U}_t = \{u_t \mid \underline{u}_t \leq u_t \leq \bar{u}_t\}$. Consider the following recursion,
\begin{equation}
\label{eq::proj_NN}
\begin{aligned}
x_t^0 &= x_t \\
x_t^{k+1}&=\max(W^{k} x_t^{k}+b^{k},0) \quad k=0, \cdots, \ell-1 \\
x_t^{\ell+1}& = \max(W^\ell x_t^{\ell}+b^\ell-\underline{u}_t,0)+\underline{u}_t \\
x_t^{\ell+2} &= -\max(\bar{u}_t-x_t^{\ell+1},0)+\bar{u}_t.
\end{aligned}
\end{equation}
Then it is not hard to show that $x_t^{\ell+2} = u_t = \proj \left( \nnc(x_t) \right)$. In other words, we have embedded the projection operator~\eqref{eq::sat_control} as two additional layers into the neural network. 

Now, we derive the QCs for the ReLU activation function, $\phi(x) = \max(0,x), \ x \in \R^d$, which is the nonlinearity used in the hidden layers of the projected neural network. Note that for $d=1$, this function is on the boundary of the sector $[0,1]$.
More precisely, we can describe it by taking the intersection of three quadratic and/or affine constraints:
\begin{equation}
\label{relu::sector}
y_{i}=\max (0, x_{i}) \Leftrightarrow y_{i} \geq x_{i}, \ y_{i} \geq 0, \ y_{i}^{2}=x_{i} y_{i},
\end{equation}
for all $i=1,\cdots,d$.
In addition, the ReLU $\phi(x)$ is also slope-restricted on $[0,1]$ with repeated nonlinearity. This allows us to write \eqref{eq::QC::slope_res} with $\alpha=0$ and $\beta=1$,
\begin{equation}
\label{relu::slope}
(y_{j}-y_{i})^{2} \leq (y_{j}-y_{i})(x_{j}-x_{i}), \ \forall i \neq j.
\end{equation}
By taking a weighted sum of the constraints \eqref{relu::sector} and \eqref{relu::slope}, we obtain the single quadratic constraint,
\begin{equation}
\begin{aligned}
\label{relu::cvx_relax}
\sum_{i=1}^{d} &\lambda_{i}(y_{i}^{2}-x_{i} y_{i})+\nu_{i}(y_{i}-x_{i})+\eta_{i} y_{i}- \\
&\sum_{i \neq j} \lambda_{i j}\left((y_{j}-y_{i})^{2}-(y_{j}-y_{i})(x_{j}-x_{i})\right) \geq 0,
\end{aligned}
\end{equation}
which holds for any multipliers $(\lambda_{i}, \nu_{i}, \eta_{i}, \lambda_{i j}) \in \R \times \R_{+}^{3}$ and $i,j \in \{1,\cdots,d\}$.
The following lemma shows how to express the above constraint in a standard-form QC \eqref{eq::QC::def}.

\begin{lemma}[QC for ReLU function]
\label{lemma::QC_relu}
The ReLU function $\phi(x) = \max(0,x): \R^{d} \rightarrow \R^{d}$ satisfies the QC defined by
\begin{equation}
\label{eq::relu::QC}
\mathcal{Q}=\left\{Q \,\middle\vert\, Q=\begin{bmatrix}
0 & T & -\nu \\
T & -2T & \nu+\eta \\
-\nu^\top & \nu^\top + \eta^\top & 0
\end{bmatrix}\right\}.
\end{equation}
with basis $[x^\top \ \phi(x)^\top \ 1]^\top$.
Here, $\eta, \nu \geq 0$ and $T \in \Sym^d_+$ is given by
\begin{equation*}
T=\sum_{i=1}^{d} \lambda_{i} e_{i} e_{i}^{\top}+\sum_{i=1}^{d-1} \sum_{j>i}^{d} \lambda_{i j}(e_{i}-e_{j})(e_{i}-e_{j})^{\top},
\end{equation*}
where $\lambda_{ij} \geq 0$ and $e_i \in \R^d$ has $1$ in the $i$-th entry and $0$ everywhere else. 
\end{lemma}
\begin{proof}
See \cite{fazlyab2019safety}.
\end{proof}
\vspace{8pt}

\noindent As we will see in Section \ref{sec::SDP}, the $Q$ matrix in \eqref{eq::relu::QC} will appear as a decision variable in the SDP problem.
Note that there are many ways to refine \eqref{eq::relu::QC} to yield a tighter relaxation for a specific region in the state-space, such as using interval arithmetic~\cite{fazlyab2019safety,fazlyab2019probabilistic} or linear programming  \cite{kolter2017provable}. 

%% file: contents/reach-SDP.tex
\section{Reach-SDP: Computing Forward Reachable Sets via Semidefinite Programming}
\label{sec::SDP}

In this section, we propose Reach-SDP, an optimization-based approach that uses the QC abstraction developed in the previous section to estimate the reachable set of the closed-loop system \eqref{eq::closed_loop_sys}.
%
%
Specifically, the $N$-step reachable set is estimated using the following recursive computations,
\begin{equation}
\begin{aligned}
\label{seq_SDP}
\ReachBar_{t+1}(\init) &= \operatorname{Reach\_{SDP}}\left(\ReachBar_{t}(\init)\right),
\end{aligned}
\end{equation}
for $t=0,\cdots,N-1$.
In the sequel, we discuss how to implement $\operatorname{Reach\_{SDP}}$ in detail.

\subsection{Change of Basis}
In the previous section, we have abstracted the initial set, the reachable set and the projected neural network with QCs or quadratic inequalities, each with a different basis vector.
For Reach-SDP, we unify those quadratic terms with the same basis vector:
\begin{equation}
\label{eq::seq_basis}
    [\bx^\top \ 1]^\top \! \coloneqq \! [{x_t^0}^\top \ {x_t^1}^\top \ \cdots \ {x_t^{\ell+2}}^\top \ 1]^\top \in \R^{n_x+n_n+2n_u+1},
\end{equation}
    where $n_n = \sum_{k=1}^{\ell} n_k$ is the total number of neurons 
    in the neural network.
Then, the unified QC is in the form: \begin{equation}
\label{eq::unified_QC}
\begin{bmatrix}
\bx \\
1
\end{bmatrix}^{\top} M 
\begin{bmatrix}
\bx \\
1
\end{bmatrix} \geq 0,
\end{equation}
where $M \in \Sym^{n_x+n_n+1}$.
The following lemma shows how to change the basis of a QC by a congruence transformation.

\begin{lemma}
\label{lemma::CoB}
The QC defined by $\mathcal{Q}= \left\{ Q \in \Sym^d \mid x_b^\top Q x_b \geq 0 \right\}$ with basis $x_b \in \R^d$ is equivalent to the QC defined by $\mathcal{M} = \left\{ M \in \Sym^n \mid \xi_b^\top M \xi_b \geq 0 \right\}$ with basis $\xi_b \in \R^n$, where $M = E^\top Q E$ and $E \xi_b = x_b$.
The matrix $E \in \R^{n \times d}$ is the change-of-basis matrix.
\end{lemma}

\begin{proof}
Substitute $E \xi_b = x_b$ into $x_b^\top Q x_b \geq 0$ and we have $x_b^\top E^\top Q E x_b = \xi_b^\top M \xi_b \geq 0$.
\end{proof}

\begin{proposition}
\label{prop::in}
The QC defined by $\mathcal{P}$ in \eqref{eq::QC::poly} or \eqref{eq::QC::ellips}, satisfied by the initial set $\init$, is equivalent to the QC defined by $\cMin = \left\{ \Min(P) \in \Sym^{n_x+n_n+1} \mid \Min(P) = \Ein^\top P \Ein \geq 0 \right\}$ with basis $[\bx^\top \ 1]^\top$.
The change-of-basis matrix is
\begin{equation}
\label{eq::seq_QC_input_CoB}
\Ein = 
\begin{bmatrix}
I_{n_x} &0 &\cdots &0 &0 \\
0 &0 &\cdots &0 &1
\end{bmatrix}.
\end{equation}
\end{proposition}


\begin{proposition}
\label{prop::out}
The quadratic inequalities \eqref{QI::reach_set} defined by matrices $S_i$ in \eqref{QI::reach_set_poly} or \eqref{QI::reach_set_ellips}, describing the candidate set $\ReachBar(\init)$, are equivalent to the quadratic inequalities defined by matrices $\Mout(S_i) = \Eout^\top S_i \Eout$ with basis $[\bx^\top \ 1]^\top$.
The change-of-basis matrix is
\begin{equation}
\label{eq::seq_QC_input_CoB}
\Eout = 
\begin{bmatrix}
A_t &0 &\cdots &0 &B_t &c_t \\
0   &0 &\cdots &0 &0 &1
\end{bmatrix}.
\end{equation}
\end{proposition}

\begin{proposition}
\label{prop::mid}
The QC defined by $\mathcal{Q}$ in \eqref{eq::relu::QC}, satisfied by the neural network controller $\nnc(x_t)$ in \eqref{eq::NN}, is equivalent to the QC defined by $\cMmid = \left\{ \Mmid(Q) \in \Sym^{n_x+n_n+1} \mid \Mmid(Q) = \Emid^\top Q \Emid \geq 0 \right\}$ with basis $[\bx^\top \ 1]^\top$.
The change-of-basis matrix is
\begin{equation}
\label{eq::seq_QC_input_CoB}
\Emid = 
\begin{bmatrix}
A &a \\
B &b\\
0 &1
\end{bmatrix}.
\end{equation}
where
\begin{equation}
\begin{aligned}
&A = \begin{bmatrix}
W^0 &\cdots &0 &0 &0 &0 \\
\vdots &\ddots &\vdots &\vdots &\vdots &\vdots \\
0 &\cdots &W^{\ell-1} &0  &0 &0 \\
0 &\cdots &0 &W^\ell &0 &0 \\
0 &\cdots &0 &0  &-I_{n_u} &0
\end{bmatrix}
&&a = \begin{bmatrix}
b^0 \\
\vdots \\
b^{\ell-1} \\
b^{\ell} -\underline{u}_t \\
\bar{u}_t
\end{bmatrix} \\[0.1cm]
&B = \begin{bmatrix}
0 &I_{n_1} &\cdots &0 &0 &0 \\
\vdots &\vdots &\ddots &\vdots &\vdots &\vdots \\
0 &0 &\cdots &I_{n_\ell}  &0 &0 \\
0 &0 &\cdots &0 &I_{n_u} &0 \\
0 &0 &\cdots &0  &0 &-I_{n_u}
\end{bmatrix}
&&b = \begin{bmatrix}
0 \\
\vdots \\
0 \\
-\underline{u}_t \\
\bar{u}_t
\end{bmatrix}.
\end{aligned}
\end{equation}
\end{proposition}

\begin{proof}
See Appendix \ref{pf::CoB_mid}.
\end{proof}


\subsection{Over-Approximating the One-Step Reachable Set}

In the next theorem, we state our main result for over-approximating the one-step reachable set $\Reach(\init)$ for the closed-loop system \eqref{eq::closed_loop_sys}.

\begin{theorem}[SDP for one-step reachable set]
\label{thm::seq_sdp}
Consider the closed-loop system \eqref{eq::closed_loop_sys}.
Suppose the initial set $\init$ and the projected neural network controller $\proj(\nnc(\cdot))$ satisfy the QCs defined by $\cMin$ and $\cMmid$, respectively, as in Proposition \ref{prop::in} and \ref{prop::mid}.
Let $\Mout(S)$ describe a candidate set $\ReachBar(\init) \subseteq \R^{n_x}$ as in Proposition \ref{prop::out} with $S=S_i \in \Sym^{n_x+1}$.
If the following LMI
\begin{equation}
\label{eq::lmi::seq_sdp}
    \Min(P) + \Mmid(Q) + \Mout(S) \preceq 0,
\end{equation}
is feasible for some matrices $(P,Q,S) \in \mathcal{P} \times \mathcal{Q} \times \Sym^{n_x+1}$, then $\Reach(\init) \subseteq \ReachBar(\init)$.
\end{theorem}

\begin{proof}
Since the initial set $\init$ satisfies the QC defined by $\cMin$, we have,
\begin{equation}
\label{eq::lmi::seq_sdp::pf::in}
    \begin{bmatrix} \bx \\ 1 \end{bmatrix}^\top \Min \begin{bmatrix} \bx \\ 1 \end{bmatrix} \geq 0,
\end{equation}
for all $x_0 \in \init$.
Similarly, the projected neural network controller $\proj(\nnc(\cdot))$ satisfying the QC defined by $\cMmid$ implies that,
\begin{equation}
\label{eq::lmi::seq_sdp::pf::mid}
    \begin{bmatrix} \bx \\ 1 \end{bmatrix}^\top \Mmid \begin{bmatrix} \bx \\ 1 \end{bmatrix} \geq 0, \quad \forall \bx \in \R^{n_x+n_n}.
\end{equation}
By left- and right-multiplying both sides of \eqref{eq::lmi::seq_sdp} by $[\bx^\top \ 1]$ and $[\bx^\top \ 1]^\top$, the basis vector in \eqref{eq::seq_basis}, we have
\begin{equation}
\label{eq::lmi::seq_sdp::pf::S_procedure}
    \begin{bmatrix} \bx \\ 1 \end{bmatrix}^\top \Min \begin{bmatrix} \bx \\ 1 \end{bmatrix} + \begin{bmatrix} \bx \\ 1 \end{bmatrix}^\top \Mmid \begin{bmatrix} \bx \\ 1 \end{bmatrix} + \begin{bmatrix} \bx \\ 1 \end{bmatrix}^\top \Mout \begin{bmatrix} \bx \\ 1 \end{bmatrix}
    \leq 0.
\end{equation}
The first two quadratic terms in \eqref{eq::lmi::seq_sdp::pf::S_procedure} are nonnegative for any $x_0 \in \init$ by \eqref{eq::lmi::seq_sdp::pf::in} and \eqref{eq::lmi::seq_sdp::pf::mid}, respectively.
Consequently, the last quadratic term must be nonpositive for all $x_0 \in \init$,
\begin{equation}
\label{eq::lmi::seq_sdp::pf::out}
    \begin{bmatrix} \bx \\ 1 \end{bmatrix}^\top \Mout \begin{bmatrix} \bx \\ 1 \end{bmatrix} \leq 0.
\end{equation}
By Proposition \ref{prop::out}, the above condition is equivalent to
\begin{equation}
\label{eq::lmi::seq_sdp::pf::out_2}
    \begin{bmatrix} y \\ 1 \end{bmatrix}^\top S \begin{bmatrix} y \\ 1 \end{bmatrix} \leq 0,
\end{equation}
for all $\bar{y} \in \{ y \mid y = f_\pi(x_0), \ x_0 \in \init \} = \Reach(\init)$.
Recall from \eqref{QI::reach_set} that the set of all points $y$ that satisfies \eqref{eq::lmi::seq_sdp::pf::out_2} is the candidate set $\ReachBar(\init)$.
Therefore, we conclude that $\ReachBar(\init)$ must be a superset of the exact one-step reachable set $\Reach(\init)$, i.e. $\Reach(\init) \subseteq \ReachBar(\init)$.
\end{proof}


\subsection{Minimum-Volume Approximate Reachable Set}
\label{subsec::min-vol}
Theorem \ref{thm::seq_sdp} provides a sufficient condition for over-approximating the one-step reachable set $\Reach(\init)$.
Now we can use this result to reformulate problem \eqref{eq::min_vol}, which finds a minimum-volume approximate reachable set $\ReachBar(\init)$.

\vspace{1pt}

If the approximate reachable set is parametrized by a polytope as in \eqref{QI::reach_set_poly}, it is difficult to find a minimum-volume polytope directly.
However, given a matrix $A \in \R^{m \times n_x}$ that describes the facets of the polytope, we can solve the following SDP problem,
\begin{equation}
\label{sdp::poly}
\displaystyle\minimize_{P \in \mathcal{P}, \ Q \in \mathcal{Q}, \ b_i \in \R} b_i \quad \text{subject to \eqref{eq::lmi::seq_sdp}}, 
\end{equation}
for all $i = 1, \cdots, m$.
For each fixed facet $a_i^\top \in \R^{n_x}$ in $A$, SDP \eqref{sdp::poly} finds the tightest halfspace $\{ y \mid a_i y \leq b_i\}$ that contains $\Reach(\init)$.
Finally, the polytopic approximate reachable set is given by the intersection of those halfsapces, i.e. $\ReachBar(\init) = \{ y \mid Ay \leq b \}$, as in Section \ref{sec::eq::QC::reach_set}.

\vspace{1pt}

If the approximate reachable set is parametrized by an ellipsoid as in \eqref{QI::reach_set_ellips}, we can easily obtain a minimum-volume ellipsoid that encloses $\Reach(\init)$ by solving,
\begin{equation}
\label{sdp::ellips}
\displaystyle\minimize_{ \substack{P \in \mathcal{P}, \ Q \in \mathcal{Q}, \\[1pt] A \in \Sym^{n_x}, \ b \in \R^{n_x}}} -\operatorname{log \ det}(A) \quad \text{subject to \eqref{eq::lmi::seq_sdp}}. 
\end{equation}
Note that \eqref{eq::lmi::seq_sdp} is not convex in $A$ and $b$.
Nonetheless we can find a convex constraint equivalent to \eqref{eq::lmi::seq_sdp} using Schur complement.
See~\cite{fazlyab2019probabilistic} for detail.

\begin{remark}
In~\cite{dutta2017output}, an MILP-based approach is proposed for estimating forward reachable sets of dynamical systems in closed-loop with neural network controllers.
If the facets are given with the same directions as the ones of the true reachable sets, then the estimated reachable sets are exact.
However, this method only works for polytopic reachable sets and does not scale well with the size of the neural network, the volume, and the dimension of the initial set.
\end{remark}

%% file: contents/experiments.tex
\section{Numerical Experiments}
In this section, we demonstrate our approach with two application examples.
The controllers used to generate training data were implemented in YALMIP~\cite{lofberg2004yalmip}.
All neural network controllers were trained with ReLU activation functions and the Adam algorithm in PyTorch.
We used MATLAB, CVX~\cite{grant2009cvx} and Mosek~\cite{mosek} to solve the Reach-SDP problems.

\subsection{Double Integrator}
\label{sec::exp::DI}
We first consider a double integrator system
\begin{equation}
\label{eq::DI}
x_{t+1}=\underbrace{\begin{bmatrix} 1 &1 \\ 0 &1\end{bmatrix}}_A x_t+\underbrace{\begin{bmatrix} 0.5 \\ 1\end{bmatrix}}_B u_t
\end{equation}
discretized with sampling time $t_s = 1$s and subject to state and input constraints, $\avoid^\complement = [-5,5] \times [-1,1]$ and $\mathcal{U} = [-1,1]$, respectively.
We implemented a standard linear MPC following~\cite{borrelli2017predictive} with a prediction horizon $N_\mpc = 10$, weighting matrices $Q = I_2$, $R = 1$, the terminal region $\mathcal{O}^{LQR}_{\infty}$ and the terminal weighting matrix $P_{\infty}$ synthesized from the discrete-time algebraic Riccati equation.
The MPC is designed as a stabilizing controller which steers the system to the origin while satisfying the constraints.
We then use the MPC to generate 2420 samples of state and input pairs $(x, \pi_\mpc(x))$ for learning.
The neural network has 2 hidden layers with 10 and 5 neurons, respectively.
Our goal is to verify if all initial states in $\init = [2.5,3] \times [-0.25,0.25]$ can reach the set $\goal = [-0.25,0.25] \times [-0.25,0.25]$, a region near the origin, in $N = 6$ steps while avoiding $\avoid$ at all times.
We computed the minimum-volume polytopic approximate reachable sets $\ReachBar^1(\init),\cdots,\ReachBar^6(\init)$ using the Reach-SDP introduced in Section \ref{sec::SDP}.
The matrix that determines the direction of the facets is chosen as
\begin{equation*}
    A_{\text{in}} = \begin{bmatrix} 
       1 &-1 &0 &0 &1 &-1 &1 &-1
    \\ 0 &0 &1 &-1 &-1 &1 &1 &-1\end{bmatrix}^\top
\end{equation*}
As shown in Figure \ref{fig::DI_1}, our approach yielded a tight outer-approximation of the reachable sets and successfully verified the safety properties sought.

\begin{figure}[!hbtp]
    \centering
    \includegraphics[width=0.7\columnwidth]{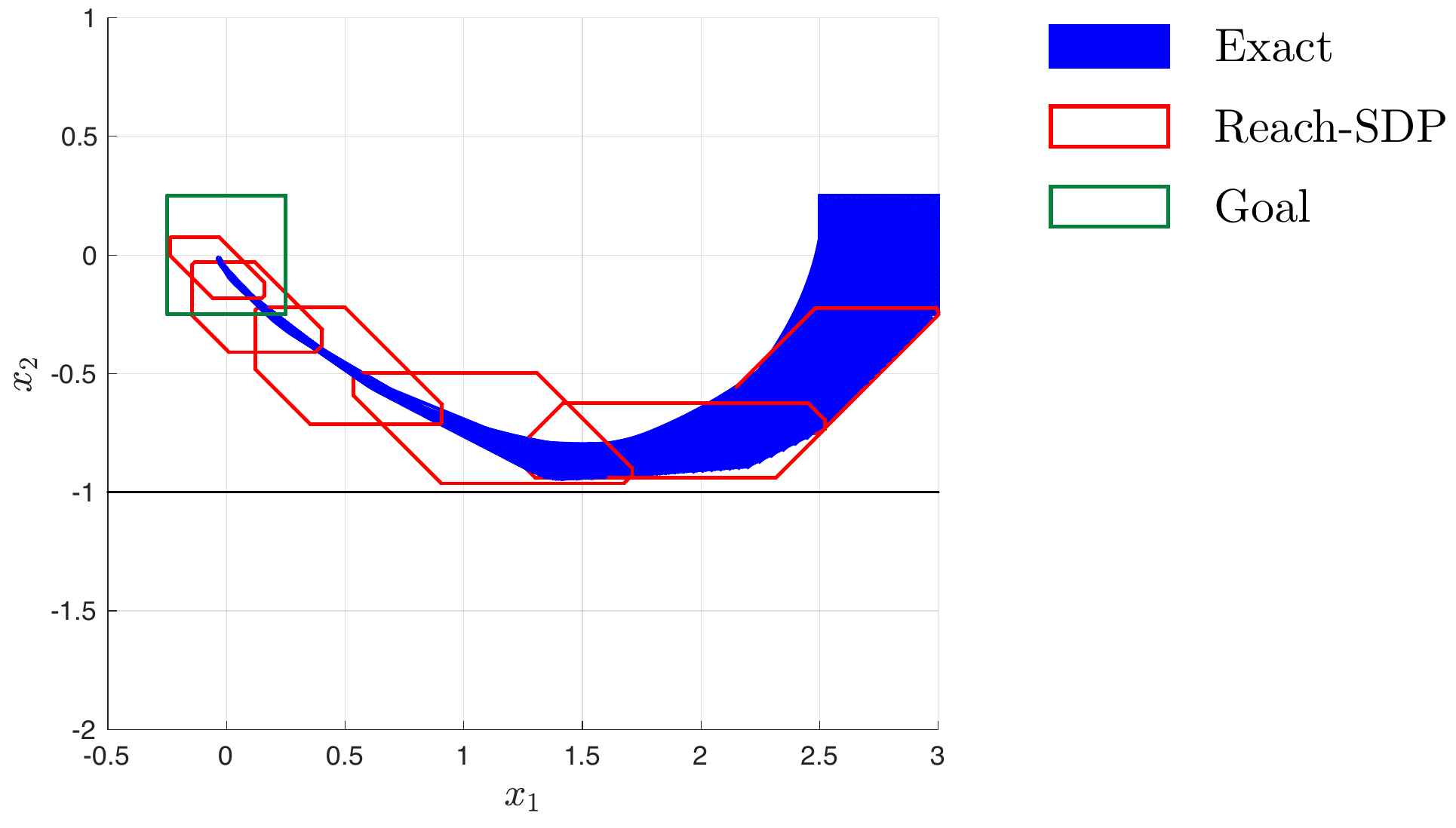}
    \caption{\label{fig::DI_1}  Illustration of the exact reachable sets (blue), the over-approximated reachable sets computed by Reach-SDP (red) and the goal set (green) of the double integrator.
    The solid black line represents the state constraint $x_2 \geq -1$.} 
\end{figure}



\subsection{6D Quadrotor}
\label{sec::exp::quad}

In the second example, we apply Reach-SDP to the 6D quadrotor model in \cite{ARCH19:Verification_of_Closed_loop_Systems}.
In order to have a linear model as in \eqref{eq::LTV}, we rewrite the nonlinear quadrotor dynamics in \cite{ARCH19:Verification_of_Closed_loop_Systems} as follows,
\begin{equation}
\label{eq::quad}
\dot{x}=\underbrace{\begin{bmatrix}
0_{3 \times 3} & I_3 \\
0_{3 \times 3} & 0_{3 \times 3}
\end{bmatrix}}_{A} x+ \underbrace{\begin{bmatrix}
0 & 0 & 0 \\
0 & 0 & 0 \\
0 & 0 & 0 \\
g & 0 & 0 \\
0 & -g & 0 \\
0 & 0 & 1
\end{bmatrix}}_{B} \underbrace{\begin{bmatrix}
\tan (\theta) \\
\tan (\phi) \\
\tau
\end{bmatrix}}_{u} +\underbrace{\begin{bmatrix}
0_{5 \times 1} \\
-g
\end{bmatrix}}_{c}
\end{equation}
where $g$ is the gravitational acceleration, the state vector $x = [p_x,p_y,p_z,v_x,v_y,v_z]^\top$ include positions and velocities of the quadrotor in the 3D space and the control vector $u$ is a function of $\theta$ (pitch), $\phi$ (roll) and $\tau$ (thrust), which are the control inputs of the original model.
The control task is to steer the quadrotor to the origin while respecting the state constraints $\avoid^\complement = [-5,5] \times [-5,5] \times [-5,5] \times [-1,1] \times [-1,1] \times [-1,1]$ and the actuator constraints $[\theta, \phi, \tau]^\top \in [-\pi/9,\pi/9] \times [-\pi/9,\pi/9] \times [0,2g]$.
We implemented a nonlinear MPC with a prediction horizon $N_\mpc = 30$, a least-squares objective function that penalizes both states and inputs with weighting matrices $Q = I_6$, $R = I_4$ and the terminal constraint $x_{N_\mpc} = 0$.
We used the original nonlinear dynamics in \cite{ARCH19:Verification_of_Closed_loop_Systems} as the prediction model in MPC, which is discretized with a sampling time $t_s = 0.1$s using the Runge–Kutta 4th order method.
The nonlinear MPC problems were solved using SNOPT~\cite{GilMS05}.
A total number of 4531 feasible samples of state and input pairs $(x, \pi_\mpc(x))$ were generated and used to train a neural network with 2 hidden layers and 32 neurons in each layer.
The initial set is given as an ellipsoid $\init = \mathcal{E}(q_{0}, Q_{0})$, where $q_{0}=[4.7 \ 4.7 \ 3 \ 0.95 \ 0 \ 0]^{\top}$ is the center and $Q_{0}=\operatorname{diag}(0.05^{2}, 0.05^{2}, 0.05^{2}, 0.01^{2}, 0.01^{2}, 0.01^{2})$ is the shape matrix.
Here, we want to verify if all initial states in $\init$ can reach the set $\goal = [3.7,4.1] \times [2.5,3.5] \times [1.2,2.6]$, which is defined in the $(p_x,p_y,p_z)$-space, in $t = 1$ second subject to the state and input constraints. 
We approximated the ellipsoidal forward reachable sets $\ReachBar^1(\init),\cdots,\ReachBar^{10}(\init)$ using the Reach-SDP.
The resulting approximate reachable sets are plotted in Figure \ref{fig::quad_xy} and \ref{fig::quad_xz} which shows that our method is able to verify the given reach-avoid specifications.

\begin{figure}[!hbtp]
    \centering
    \includegraphics[width=0.7\columnwidth]{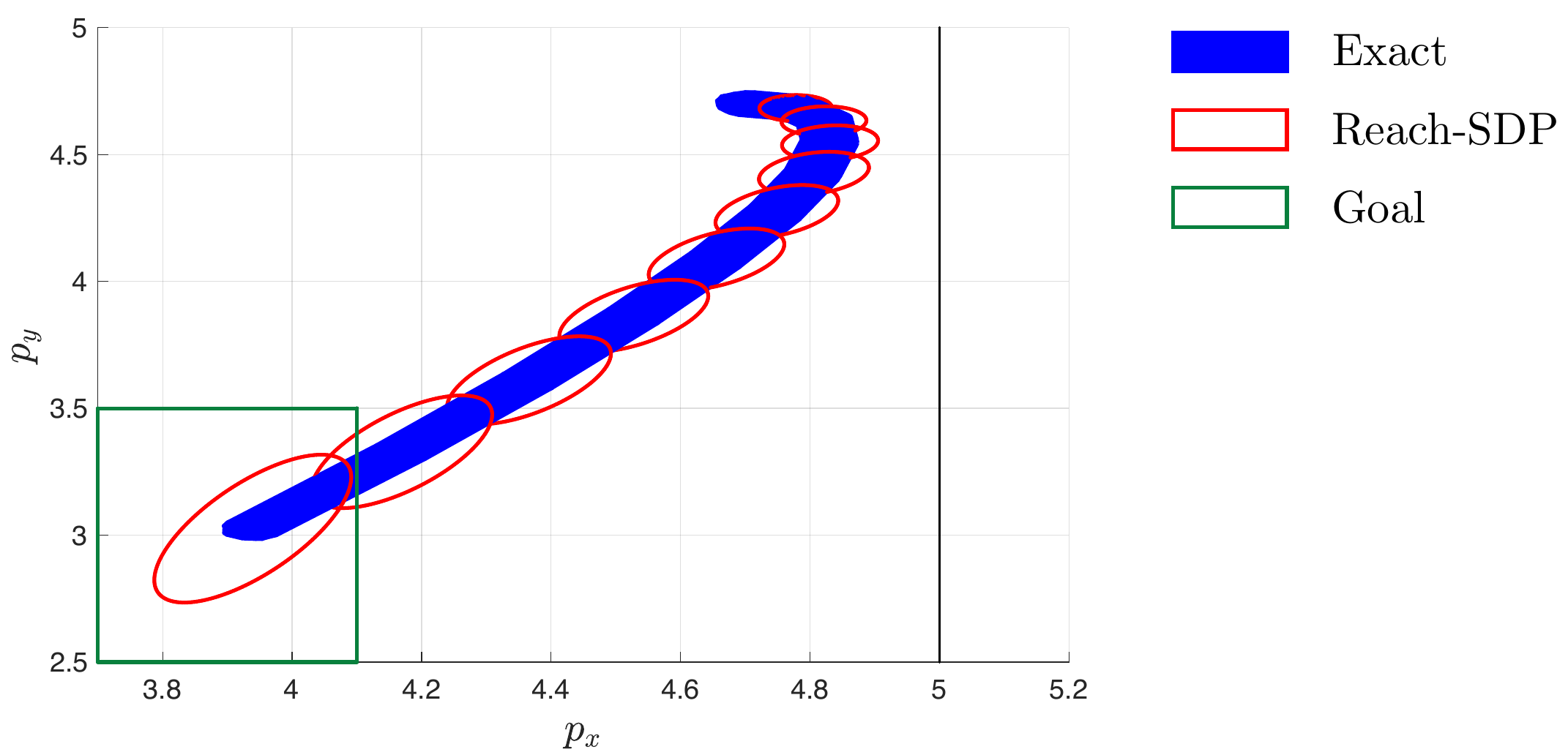}
    \caption{\label{fig::quad_xy}  Illustration of the exact and approximate reachable sets of the quadrotor system in the $(p_x,p_y)$-space. The solid black line represents the state constraint $p_x \leq 5$.} 
\end{figure}

\begin{figure}[!hbtp]
    \centering
    \includegraphics[width=0.7\columnwidth]{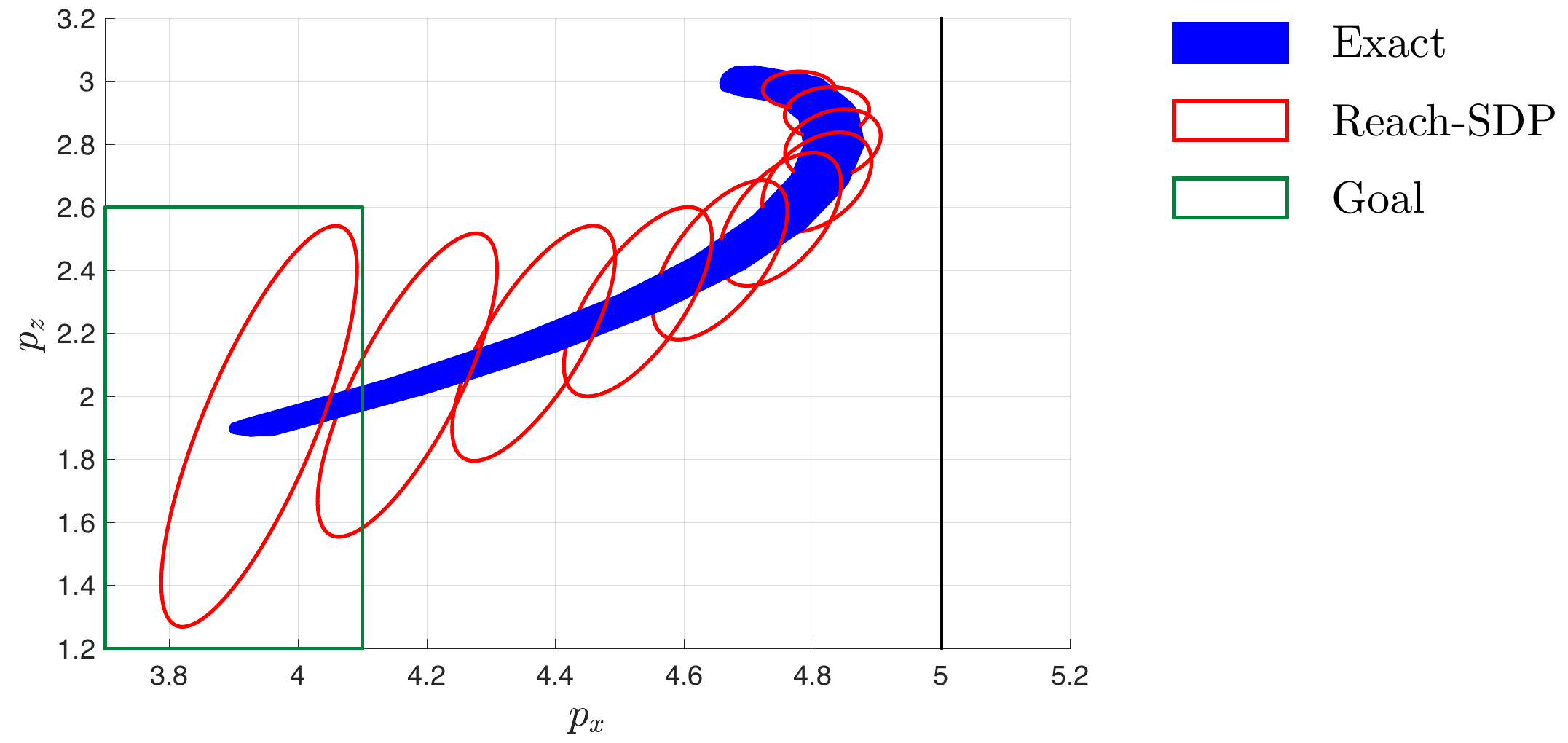}
    \caption{\label{fig::quad_xz}  Illustration of the exact and approximate reachable sets of the quadrotor system in the $(p_x,p_z)$-space.} 
\end{figure}

%% file: contents/conclusion.tex
\section{Conclusions}
In this paper, we propose the first convex-optimization-based reachability analysis method for linear systems in feedback interconnection with neural network controllers. Our approach relies on abstracting the nonlinear components of the closed-loop system by quadratic constraints. Then we show that we can compute the approximate reachable sets via semidefinite programming. Future work includes extending the current approach to incorporate nonlinear dynamics and to approximate backward reachable sets, which is useful for certifying invariance properties.

%% file: contents/appendix.tex
\appendix
\section{Proof of Proposition \ref{prop::mid}}
\label{pf::CoB_mid}
Assuming the same activation function for all neurons throughout the entire network, we can write \eqref{eq::proj_NN} compactly as
\begin{equation}
\label{eq::compact_NN}
    B \bx+b =\phi(A \bx+a)
\end{equation}
where $\bx \in \R^{n_x+n_n+2n_u}$ is used to define the basis vector in \eqref{eq::seq_basis}.
Now, we introduce two auxiliary variables $\mathbf{y}$ and $\mathbf{z}$ such that $\mathbf{y} = A \bx+a$ and $\mathbf{z} = B \bx+b$.
By Lemma \ref{lemma::QC_relu}, the neural network in \eqref{eq::compact_NN} satisfies the QC defined by $\mathcal{Q}$ in \eqref{eq::relu::QC} with basis $[\mathbf{y}^\top \ \mathbf{z}^\top \ 1]^\top$.
Consider the congruence transformation,
\begin{equation}
     \begin{bmatrix} \mathbf{y} \\ \mathbf{z} \\ 1 \end{bmatrix} = \Emid \begin{bmatrix} \bx \\ 1 \end{bmatrix}.
\end{equation}
By Lemma \ref{lemma::CoB} the neural network \eqref{eq::proj_NN} satisfies the QC defined by $\cMmid$ with basis $[\bx^\top \ 1]^\top$, which concludes the proof.